\documentclass[conference]{IEEEtran}
\IEEEoverridecommandlockouts

\usepackage{setspace}
\usepackage{indentfirst}
\usepackage[numbers,sort&compress]{natbib}
\usepackage{xcolor}
\usepackage[pdftex]{graphicx}
\usepackage{amsmath}
\usepackage{amsthm}
\usepackage{amssymb}
\usepackage[ruled,linesnumbered]{algorithm2e}
\usepackage{algpseudocode}
\usepackage{array}
\usepackage[caption=true,font=footnotesize]{subfig}
\usepackage{fixltx2e}
\usepackage{stfloats}
\usepackage{url}
\usepackage{comment}
\usepackage{multirow}
\usepackage{makecell}
\usepackage{ulem}
\usepackage{verbatim}
\usepackage{flushend}
\allowdisplaybreaks[2]
\newtheorem{thm}{Theorem}

\newtheorem{assum}{Assumption}

\theoremstyle{definition}

\def\BibTeX{{\rm B\kern-.05em{\sc i\kern-.025em b}\kern-.08em
    T\kern-.1667em\lower.7ex\hbox{E}\kern-.125emX}}

\renewcommand{\IEEEQED}{\IEEEQEDopen}

\usepackage[T1]{fontenc}
\begin{document}

\title{Optimization-as-a-Service via Multi-Agent Large Language Model for Radio Access Networks}

\author{
  \IEEEauthorblockN{
    Chaoqun You\IEEEauthorrefmark{1},
    Yueyue Dai\IEEEauthorrefmark{2},
    Xingqiu He\IEEEauthorrefmark{1},
    Yue Gao\IEEEauthorrefmark{1},
    Rahim Tafazolli\IEEEauthorrefmark{3},
    Yong Liang Guan\IEEEauthorrefmark{4}
  }
  \IEEEauthorblockA{\IEEEauthorrefmark{1}Fudan University, \IEEEauthorrefmark{2}Huazhong University of Science and Technology, \IEEEauthorrefmark{3}University of Surrey}
  \IEEEauthorblockA{\IEEEauthorrefmark{4}Nanyang Technological University}
}

\maketitle

\begin{abstract}
The physical resource block (PRB) allocation in Radio Access Networks (RANs) traditionally relies on case-by-case manual problem construction or, more recently, learning-based artificial intelligence (AI) methods. However, the sixth-generation (6G) RAN environments confront unprecedented service diversity and exponential dynamics, featuring volatile fluctuations in active base stations (BSs), user scale, and stringent Quality-of-Service (QoS) requirements. Faced with such conditions, both manual models and standard AI algorithms remain fundamentally rigid, lacking the flexibility to adapt and self-evolve. To provide a one-size-fits-all solution, we propose treating the PRB allocation problem as an Optimization-as-a-Service (OaaS) provided by a large language model multi-agent (LLM-MA) system. This fundamentally reshapes RAN resource allocation by utilizing agents to dynamically construct optimization problems and automatically determine objectives tailored to real-time scenarios. Our closed-loop architecture, integrating scene understanding, objective generation, solver, and reflection agents, enables context-aware, self-correcting formulation. To eliminate the computational latency of iterative reflection, we introduce a one-shot reflection distillation mechanism, training a lightweight student model to directly predict refined objective parameters. We theoretically bound the performance gap of this one-shot policy. Experimental results demonstrate our framework achieves near-optimal resource allocation with ultra-low inference latency.
\end{abstract}


\section{Introduction}
As wireless communications evolve towards the 6G era, RANs must support unprecedented service diversity through emerging applications like extended reality (XR), connected autonomous vehicles (CAVs), digital twins, \textit{etc} \cite{nguyen2021security}. These advanced applications introduce stringent, heterogeneous QoS and security requirements that, coupled with volatile spatio-temporal traffic bursts and expanding attack surfaces, 6G RANs are subjected to extreme exponential dynamics. Consequently, efficient PRB allocation in such a hyper-dynamic and vulnerable ecosystem becomes profoundly challenging, requiring the network to continuously balance highly competing performance and security metrics without compromising service continuity.

Conventional PRB allocation relies heavily on case-by-case \textit{manual} problem construction based on specific assumed scenarios. More recently, AI approaches, particularly Deep Reinforcement Learning (DRL) \cite{arulkumaran2017deep}, have been actively deployed to bypass explicit problem formulation and directly output allocation decisions. However, this bipartite paradigm is ill-suited for modern volatile RANs where the number of active BSs, user equipments (UEs), and QoS demands fluctuate drastically. A problem formulated for one scenario quickly becomes invalid when the network shifts. Manually reconstructing models for all possible states is intractable. Concurrently, DRL policies sacrifice mathematical interpretability and struggle to generalize across heterogeneous network topologies without extensive retraining.

To address these limitations, we propose a paradigm shift from manual formulation and black-box mapping to autonomous and interpretable problem synthesis. Specifically, we propose an Optimization-as-a-Service (OaaS) framework driven by LLM-MA systems~\cite{he2025llm}, which utilizes agents to dynamically construct optimization problems and formulate objectives tailored to instantaneous scenarios. The framework employs a closed-loop architecture with four components: a \texttt{Scene Understanding Agent} to extract semantic representations from raw RAN states, an \texttt{Objective Generation Agent} to synthesize context-aware optimization parameters, a \texttt{Solver Agent} for PRB allocation, and a \texttt{Reflection Agent} to iteratively refine the formulated problem against QoS constraints.

Iterative reflection enhances adaptability but introduces computational overhead unacceptable for millisecond-level RAN scheduling. Thus, we further propose a one-shot objective generation mechanism via reflection distillation. Using the multi-agent pipeline as a teacher, a lightweight student model is trained to directly predict high-quality objective parameters in a single forward pass. Theoretically, we prove that the performance gap between the one-shot policy and the reflection-refined solution is bounded.

The main contributions of this paper are summarized as follows:

\begin{itemize}
    \item We propose to treat the 6G RAN PRB allocation optimizaion as a service provided by LLM-MA to provide a one-size-fits-all solution for diverse RAN scenarios.
    \item We design a closed-loop LLM-MA framework, OaaS, with four agents to deliver a fully automated, self-evolving PRB sharing solution (Section II).
    \item We further propose a one-shot OaaS mechanism via reflection distillation to guarantee rapid inference in practical applicationsn (Section III).
    \item We validate the effectiveness of OaaS by leveraging GPT-4o~\cite{openai2024gpt4o}. The OaaS  consistently outperforms other baselines, while the one-shot model successfully reduce the decision-making delay to around 0.2s (Section VI).
\end{itemize}

\section{Optimization-as-a-Service Framework} \label{secII}

\begin{figure*}[!t]
\centering
  \includegraphics[width=0.9\linewidth]{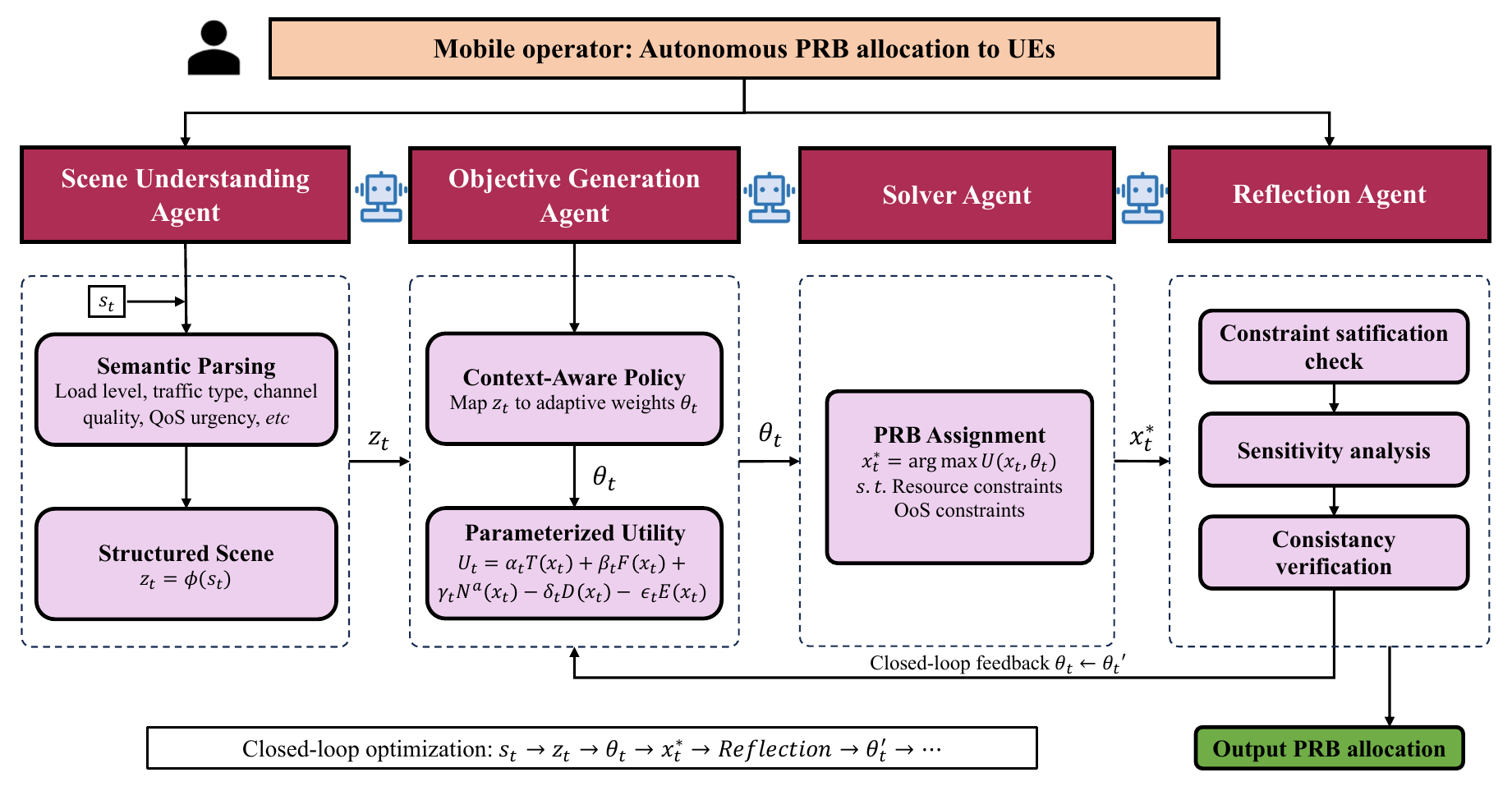}
  \caption{The OaaS Framework}
  \label{fig:overview}
\end{figure*}

\subsection{System Overview} 
We consider a RAN where a BS allocates a set of PRBs $\mathcal{K}$ to a dynamic set of UEs $\mathcal{U}$. The system operates in time slots indexed by $t$. At each time slot $t$, the network state evolves due to user mobility, traffic variations, and channel fluctuations. The target is to determine a PRB allocation strategy that adapts to time-varying network conditions and heterogeneous UE QoS requirements. To avoid manually designing optimization objectives for each scenario, we propose an LLM-MA framework, OaaS, that automatically constructs and refines the optimization problem. As shown in Fig.~\ref{fig:overview}, the framework consists of four agents: \texttt{Scene Understanding Agent}, \texttt{Objective Generation Agent}, \texttt{Solver Agent}, and \texttt{Reflection Agent}. Next we will introduce the four agents in detail.

\subsection{The Scene Understanding Agent}

The \texttt{Scene Understanding Agent} processes the raw RAN state and extracts a structured representation. At each time slot $t$, the network state is represented as $s_t=\{N_t, \mathbf{q}_t,\mathbf{h}_t,\mathbf{d}_t\}$, where $N_t=|\mathcal{U} (t)|$ is the number of users, $\mathbf{q}_t$ is the QoS requirement vector (e.g., latency, reliability), $\mathbf{h}_t$ is the channel state information (e.g., CQI), and $\mathbf{d}_t$ is the traffic demand vector. The state $s_t$ is first processed by the Scene Understanding agent, which extracts a structured semantic representation, $z_t=\Phi(s_t)$, where $z_t$ encodes high-level features such as load level, latency sensitivity, and fairness requirements.

\subsection{The Objective Generation Agent}

The \texttt{Objective Generation Agent} takes the semantic representation $z_t$ and produces parameters for a parameterized optimization objective. Let $x_{i,k}(t)\in\{0,1\}$ denote whether PRB $k$ is assigned to user $i$ during time slot $t$, where $i\in\mathcal{U}, k\in\mathcal{K}$. The allocation matrix is $\{\mathbf{x}_{i,k}(t)|i\in\mathcal{U},k\in\mathcal{K}\}$, subject to $\sum_{i\in\mathcal{U}(t)} x_{i,k}(t)\leq 1$.

Instead of manually defining a fixed objective, we adopt a parameterized utility function:
\begin{subequations}
    \begin{align*}
        U(\mathbf{x}_t;\mathbf{\theta}_t) = & \alpha_t T(\mathbf{x}_t) + \beta_t F(\mathbf{x}_t) + \gamma_t N^a(\mathbf{x}_t) \\
        & - \delta_t D(\mathbf{x}_t) - \epsilon_t E(\mathbf{x}_t),
    \end{align*}
\end{subequations}
where $T(\cdot)$, $F(\cdot)$, $N^a(\cdot)$, $D(\cdot)$ and $E(\cdot)$ represent the system throughput, fairness index, admitted number of UEs, latency and energy consumption, respectively. The parameters $\mathbf{\theta}_t = (\alpha_t,\beta_t,\gamma_t, \delta_t, \epsilon_t)$ determine the relative importance of these metrics, with each parameter taking a value between 0 and 1.

The objective generation is modeled as a policy $\theta_t=f_{\theta}(z_t)$, which maps the semantic state representation to objective parameters. Instead of requiring explicit supervision for optimal parameters, we adopt a \textit{preference learning} approach~\cite{xu2026visionreward}. Specifically, the agent learns from pairwise comparisons: 
\begin{equation}
    (s,\mathbf{\theta}^a) \succ (s,\mathbf{\theta}^b),
\end{equation}
indicating that parameter set $\mathbf{\theta}^a$ leads to better system performance under state $s$. The preference model is trained by minimizing a pairwise ranking loss:
\begin{equation}
    \mathcal{L} = -\log \sigma (P(s,\mathbf{\theta}^a)-P(s,\mathbf{\theta}^b)),
\end{equation}
where $P(\cdot)$ is performance score derived from system metrics. $\sigma(\cdot)$ denotes the sigmoid function, maps the difference between two candidate configurations into a probability value. This formulation enables the model to learn a consistent ranking over candidate objective parameters. Importantly, the model does not learn a fixed parameter vector; instead, it learns a state-dependent mapping from the RAN state to the objective parameters, allowing adaptive and context-aware objective generation under dynamic network conditions.

\subsection{Solver Agent}

The \texttt{Solver Agent} computes the optimal PRB allocation based on the generated objective. That is, given the generated parameters $\mathbf{\theta}_t$, the \texttt{Solver Agent} computes the optimal PRB allocation:
\begin{equation}
    \mathbf{x}^*_t=\arg\max_{\mathbf{x}_t} U(\mathbf{x}_t;\mathbf{\theta}_t),
\end{equation}
subject to resource constraints. The solver can be implemented using various optimization techniques, such as integer linear programming, heuristic algorithms, or reinforcement learning-based solvers, depending on the complexity of the problem and the required solution quality.

\subsection{Reflection Agent}

Finally, the \texttt{Reflection Agent} evaluates whether the selected parameters $\theta_t$ are aligned with system requirments and provides feedback to refine $\theta_t$ iteratively. The reflection process consists of three componets: (i) constraint satisfaction check, which can be formalized as follows,
\begin{subequations}
    \begin{align*}   
        & T(\mathbf{x}_t^*) \geq  T_{\text{min}}, F(\mathbf{x}_t^*) \geq  F_{\text{target}}, N^a  (\mathbf{x}_t^*) \geq N_t, \\
        & D(\mathbf{x}_t^*) \leq  D_{\text{target}}, E(\mathbf{x}_t^*) \leq  E_{\text{target}}.
    \end{align*}
\end{subequations}
(ii) sensitivity analysis, that if $\frac{\partial T}{\partial \alpha}>0$, $\frac{\partial F}{\partial \beta}>0$, and $\frac{\partial D}{\partial \gamma}>0$, and (iii) consistancy verfication. The agent evaluates the robustness of the solution under small perturbations of the objective parameters, ensuring that the generated objective is self-consistent. If any of the above conditions are violated, the \texttt{Reflection Agent} provides feedback to adjust $\theta_t$, forming a closed-loop refinement process.

\subsection{Closed-Loop Optimization Framework}

The overall system operates as:
\begin{equation}
    s_t \rightarrow z_t \rightarrow  \theta_t \rightarrow  x_t^* \rightarrow  Reflection \rightarrow \theta_t'.
\end{equation}
This iterative process enables adaptive objective construction, automatic error correction, and robust performance under dynamic RAN conditions.

Our goal is to jointly optimize the objective generation policy $f_\theta (\cdot)$ and the resulting PRB allocation strategy $\mathbf{x}_t$ such that the long-term system performance is maximized,
\begin{equation}
    \max_{f_{\theta}} \mathbb{E} [P(s_t, \mathbf{x}_t^*)],
\end{equation}
subject to the QoS constraints and resource limitations. The multi-agent framework provides a flexible and adaptive approach to achieve this goal, enabling the system to learn effective objective generation strategies that are tailored to the dynamic RAN environment.

\section{One-shot Objective Generation} \label{secIII}

The proposed multi-agent framework relies on an iterative loop involving objective generation and reflection-based refinement to obtain high-quality objective parameters. While this process improves robustness, it introduces non-negligible computational overhead due to repeated evaluation and adjustment. To address this issue, we propose a one-shot objective generation mechanism, which aims to directly predict high-quality objective parameters in a single forward pass, thereby eliminating the need for iterative refinement during inference.

\subsection{Reflection Distillation}

To enable one-shot objective generation, we adopt a reflection distillation strategy, where the iterative pipeline serves as a teacher model, and the one-shot generator acts as a student model. Let $\theta^{\text{ref}}$ denote the refined objective parameters obtained after the reflection process $\theta^{\text{ref}} = \mathcal{R}(s)$, where $\mathcal{R}(\cdot)$ represents the full pipeline including objective generation, solver and reflection. $\theta^{ref}$ serves as the teacher model. Meanwhile, the student model is defined as a parameterized mapping $\theta^{\text{one}} = f_{\psi }^{\text{one}} (s)$, where $f_{\psi}^{\text{one}}$ is the one-shot objective generation policy with parameters $\psi$. The goal of one-shot inference is to learn a mapping such that $\theta^{\text{one}}$ closely approximates $\theta^{\text{ref}}$, that is, $\theta^{\text{one}}\approx \theta^{\text{ref}}$. 

\subsubsection{Teacher signal generation} 

To construct the training dataset, we sample network states from a distribution $s\sim \mathcal{D}$, which captures diverse traffic patterns, user densities, and QoS requirements. For each sampled state $s_i$, the teacher pipeline is executed to obtain the corresponding refined parameters $\mathbf{\theta}_i^{ref}$. This results in a dataset:
\begin{equation}
    \mathcal{D}_{teach} = \{(s_i,\mathbf{\theta}_i^{ref})\}_{i=1}^N.
\end{equation}
The dataset $\mathcal{D}_{teach}$ serves as supervision for training the one-shot objective generation policy in the subsequent subsection. Notably, the teacher pipeline is executed offline, and its computational complexity does not affect the online deployment phase.

\subsubsection{Student model training}

After generating the teacher signals, we train the one-shot model by minimizing the discrepancy between the student outputs and the teacher outputs. That is, given the dataset $\mathcal{D}_{teach}$, the student policy is optimized by:
\begin{equation}
    \mathcal{L}_{distill} = \|\theta^{one}(s)- \theta^{ref}(s)\|^2.      
\end{equation}
To further align the learned policy with system-level performance, we optionally incorporate a reward-based regularization term:
\begin{equation}
    \mathcal{L}_{reward} = -P(s,\theta^{one}(s)).
\end{equation}
The overall training objective for the one-shot model is a weighted combination:
\begin{equation}
    \mathcal{L} = \mathcal{L}_{distill} + \lambda \mathcal{L}_{reward},
\end{equation}
where $\lambda$ controls the trade-off between distillation fidelity and performance optimization. By minimizing this loss, the one-shot model learns to directly predict high-quality objective parameters that approximate the refined parameters from the iterative pipeline, enabling efficient inference without sacrificing performance.

\subsection{Performance Gap Analysis}

To provide theoretical insight into the effectiveness of the proposed one-shot objective generation, we characterize the performance gap between the one-shot policy and the reflection-refined solution.

\begin{assum}
The optimal solution mapping $x^*(s,\theta)$ is Lipschitz continuous with respect to $\theta$, \textit{i.e.},
\begin{equation}
    ||x^*(s,\theta_1)-x^*(s,\theta_2)|| \leq K ||\theta_1-\theta_2||,
\end{equation}
where $K>0$ is a constant.
\end{assum}

\begin{assum}
The system-level performance function $P(s,x)$ is Lipschitz continuous with respect to the allocation $\mathbf{x}$, \textit{i.e.},
\begin{equation}
    |P(s,x_1)-P(s,x_2)| \leq M ||x_1-x_2||,
\end{equation}
where $M>0$ is a constant.
\end{assum}

\begin{thm}
Under Assumptions 1 and 2, the performance gap between the one-shot objective generation policy and the reflection-refined solution is bounded as
\begin{equation}
    |P(s,x^*(s,\theta^{one}))-P(s,x^*(s,\theta^{ref}))| \leq KM ||\theta^{one}-\theta^{ref}||.
\end{equation}
\end{thm}

\begin{IEEEproof}
The performance gap can be decomposed as follows:
\begin{subequations}
   \begin{align}
    &|P(s,x^*(s,\theta^{one}))-P(s,x^*(s,\theta^{ref}))| \\
    &\leq |P(s,x^*(s,\theta^{one}))-P(s,x^*(s,\theta^{ref}))| \\
    &\leq M ||x^*(s,\theta^{one})-x^*(s,\theta^{ref})|| \\
    &\leq KM ||\theta^{one}-\theta^{ref}||.
    \end{align} 
\end{subequations}
The first inequality follows from the Lipschitz continuity of the performance function, and the second inequality follows from the Lipschitz continuity of the optimal solution mapping. This establishes that the performance gap is linearly bounded by the distance between the one-shot parameters and the reflection-refined parameters, providing a theoretical guarantee of the one-shot objective generation approach.
\end{IEEEproof}

\subsection{One-Shot Objective Generation Algorithm}

The one-shot objective generation algorithm can be summarized in Alg.~\ref{alg:one-shot}. The algorithm consists of two main phases: teacher signal generation and student model training. In the first phase, we sample network states and execute the reflection pipeline to obtain refined objective parameters, constructing the teacher dataset. In the second phase, we iteratively train the one-shot model by minimizing the combined distillation and reward-based loss until convergence. Once trained, the one-shot model can directly predict high-quality objective parameters for new network states without requiring iterative refinement, enabling efficient and adaptive optimization in dynamic RAN environments.

\begin{algorithm}[t]
\caption{One-Shot Objective Generation via Reflection Distillation}
\label{alg:one-shot}                                                   
\KwIn{Dataset of network states $\mathcal{D} = \{s_i\}_{i=1}^N$, teacher pipeline $\mathcal{R}(\cdot)$, student model $f_{\psi}^{one}(\cdot)$, reward function $R(s,\theta)$, distillation weight $\lambda$.}
\KwOut{Trained one-shot objective generation model $f_{\psi}^{one}(\cdot)$.}
\For{each $s_i$ in $\mathcal{D}$}{
    $\theta_i^{ref} \leftarrow \mathcal{R}(s_i)$ \tcp*{Generate teacher signal}
    $\mathcal{D}_{teach} \leftarrow \mathcal{D}_{teach} \cup \{(s_i,\theta_i^{ref})\}$ \tcp*{Construct teacher dataset}
}
\While{not converged}{
    Sample a mini-batch $\{(s_j,\theta_j^{ref})\}$ from $\mathcal{D}_{teach}$ \tcp*{Sample mini-batch}
    $\theta_j^{one} \leftarrow f_{\psi}^{one}(s_j)$ \tcp*{Student model prediction}
    $\mathcal{L}_{distill} \leftarrow \frac{1}{B} \sum_j \|\theta_j^{one}-\theta_j^{ref}\|^2$ 
    $\mathcal{L}_{reward} \leftarrow -\frac{1}{B} \sum_j R(s_j,\theta_j^{one})$ 
    $\mathcal{L} \leftarrow \mathcal{L}_{distill} + \lambda \mathcal{L}_{reward}$ 
    Update $\psi$ by minimizing $\mathcal{L}$ \tcp*{Update student model}
}
\end{algorithm}   

\section{Performance Evaluation}

\subsection{Setup}
To comprehensively evaluate the OaaS framework and the one-shot reflection distillation mechanism, we develop a system-level simulator compliant with 3GPP 5G New Radio (NR) specifications~\cite{3gpp.38.104}. We model a macro-cell BS serving a dynamically varying set of UEs within a coverage radius of 250 meters. The network operates in the Sub-6 GHz (FR1) band at a carrier frequency of 3.5 GHz. With a system bandwidth of 100 MHz and a subcarrier spacing (SCS) of 30 kHz ($\mu=1$), the BS provides a total of $K = 273$ PRBs per time slot. The wireless channel is characterized by the 3GPP Urban Micro (UMi) path-loss model superimposed with Rayleigh block fading~\cite{3gpp.38.901}, and thermal noise power is set to $-114$ dBm/MHz. The LLM being called is GPT-4o.

\noindent \textbf{1) Evaluation Scenarios}: To validate the necessity and adaptability of dynamic objective construction, our simulations are conducted across four highly dynamic and challenging operational scenarios,
\begin{itemize}
    \item \textbf{Scenario 1: Overloaded State (Congestion)}: The UE population surges to a large $N$ (e.g., $N=300$), causing total service demand to vastly exceed the capacity of the 273 PRBs. This scenario tests the agent’s ability to prevent starvation constraints while maximizing the number of sustainably served users ($\alpha \uparrow$, $\gamma \uparrow$).
    \item \textbf{Scenario 2: Underloaded State (Resource Conservation)}: The cell enters a sparse state (e.g., $N=10$). Rather than aggressively maximizing throughput via unnecessary resource consumption, the evaluation focuses on whether the system can reformulate the objective to minimize active PRBs while satisfying base QoS limits, effectively supporting energy-saving operations. ($\alpha \downarrow$, $\epsilon \uparrow$).
    \item \textbf{Scenario 3: Heterogeneous Traffic Burst}: The cell, initially serving background eMBB traffic, experiences a sudden influx of ultra-reliable low-latency communications (URLLC) packets (e.g., constituting 30$\%$ of the active traffic). This evaluates the agent's capability to instantly heavily penalize latency ($\delta \uparrow$).
\end{itemize}

\noindent \textbf{2) Baselines}: We compare our proposed OaaS framework against the following baselines,
\begin{itemize}
    \item \textbf{Static Max-Rate (Greedy)}: A fixed optimization formulation heavily weighting sum-throughput ($\alpha=1, \beta=0, \gamma=0, \delta=0, \epsilon=0$).
    \item \textbf{Static Proportional Fair (PF)}: A classic fixed heuristic promoting baseline sum-log throughput fairness ($\alpha=1, \beta=1, \gamma=0, \delta=0, \epsilon=0$).
    \item \textbf{DRL}: A Proximal Policy Optimization (PPO)-based black-box algorithm that bypasses objective formulation entirely and directly maps network states to PRB allocation matrices.
    \item \textbf{Iterative OaaS (Teacher)}: Our proposed multi-agent framework utilizing the iterative reflection loop to dynamically generate $\theta_t$. This represents the performance upper bound.
    \item \textbf{One-Shot OaaS (Student)}: Our proposed reflection-distilled lightweight model, aiming for equal performance with deterministic inference latency.
\end{itemize}

\subsection{Effectiveness of the OaaS Framework}

\begin{figure*}[!t]
  \centering
  \subfloat{
      \includegraphics[width=3.2in]{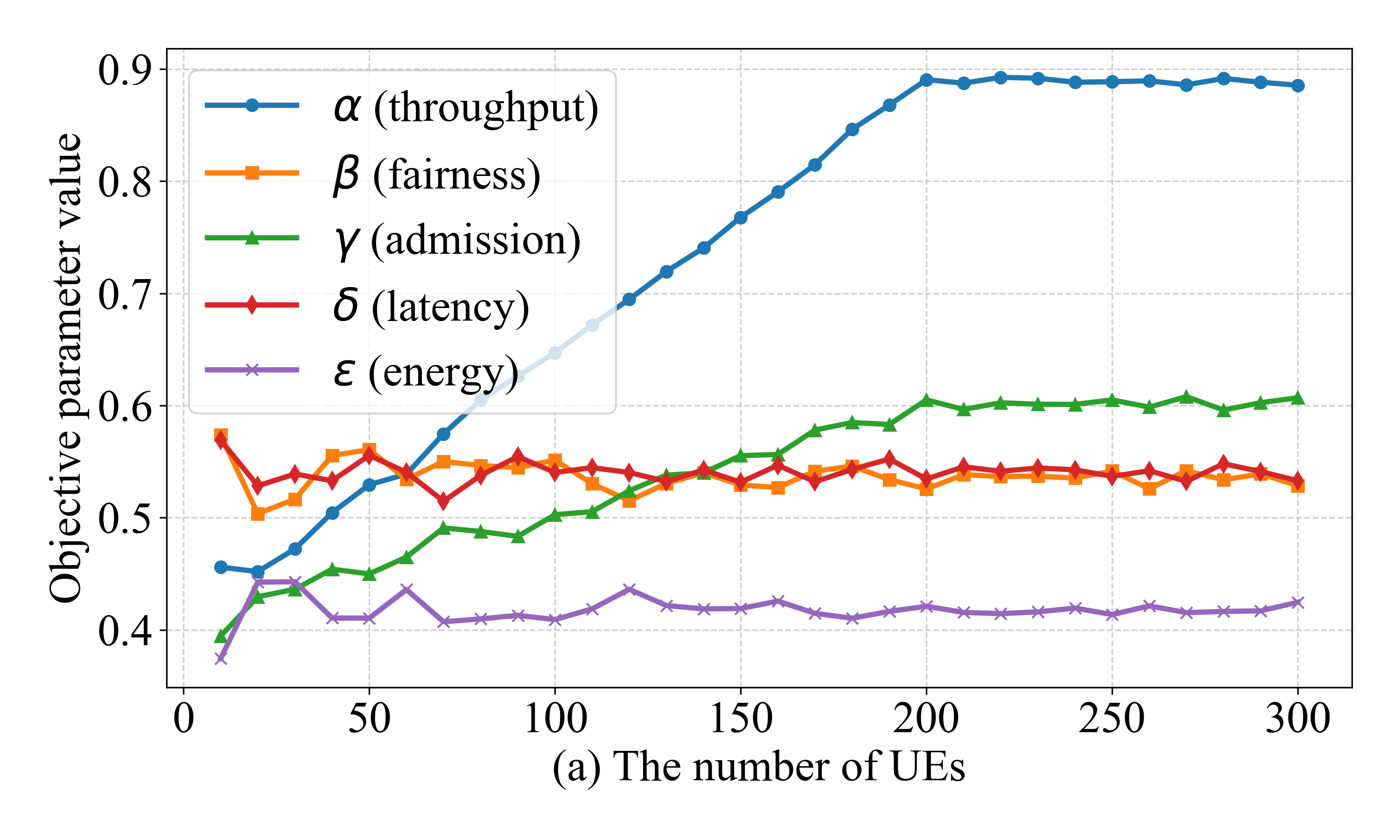}
      \label{fig:params:a}}
      \hfil
  \subfloat{
      \includegraphics[width=3.2in]{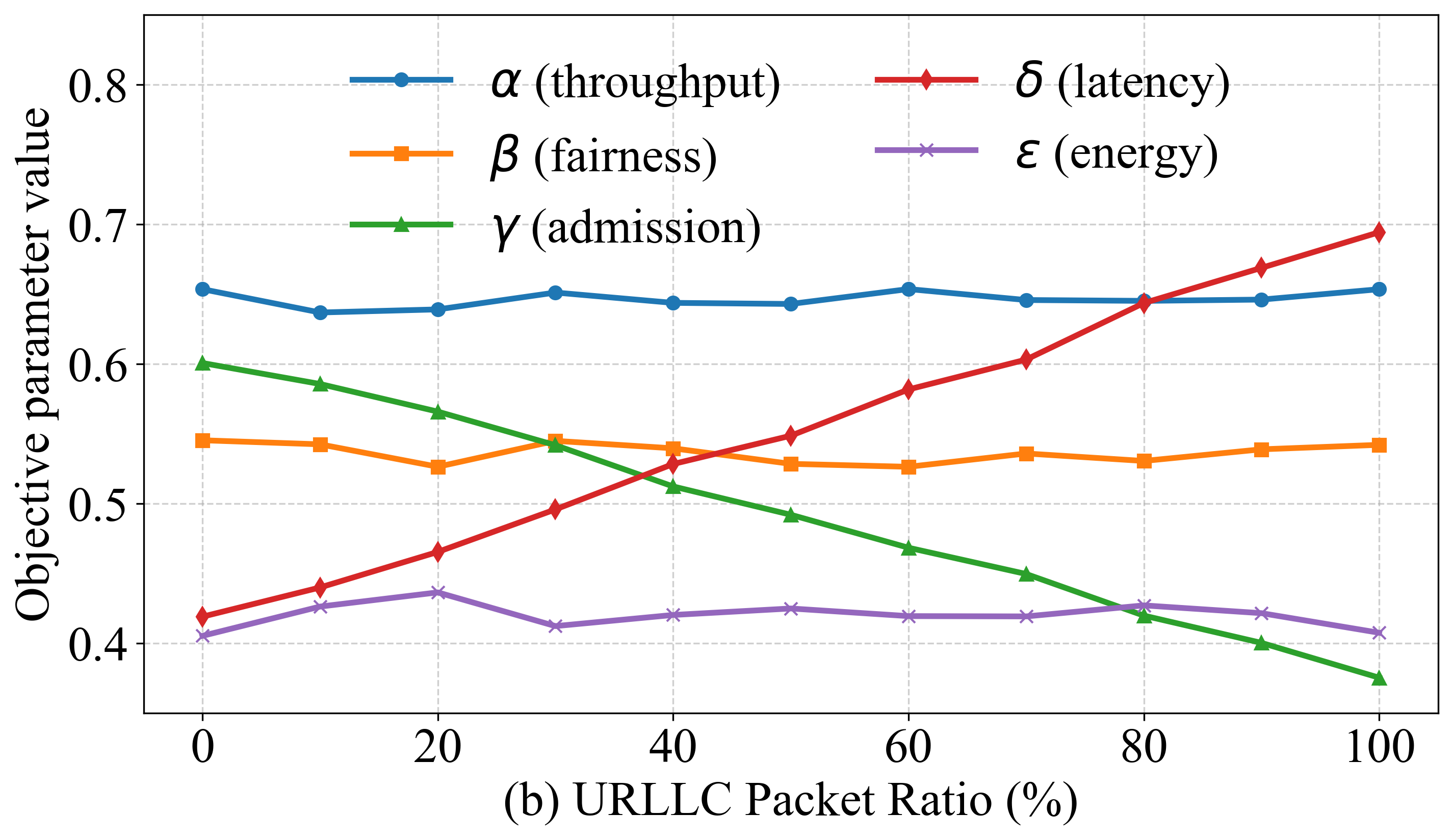}
      \label{fig:params:b}}
  \caption{The performance of the OaaS framework under different scenarios, where (a) shows the objective parameter values under different number of users (Scenario 1 and 2), and (b) shows the objective parameter values under different URLLC packet ratios (Scenario 3).}
  \label{fig:params}
\end{figure*}

Fig. \ref{fig:params}(a) illustrates the adaptive behavior of the OaaS LLM-MA framework in response to varying user densities, spanning from underloaded (e.g., 10 UEs) to overloaded (e.g., 190 UEs) scenarios. In lightly loaded networks, the framework intelligently reduces the throughput weight $\alpha$ (from approximately 0.46 to 0.45) and increases the energy conservation penalty $\epsilon$ (peaking at 0.44), thereby minimizing unnecessary PRB utilization while ensuring QoS compliance. This energy-efficient strategy aligns with resource conservation objectives in sparse conditions. Conversely, as UE count escalates, $\alpha$ rises sharply to 0.87, and the admission weight $\gamma$ increases from 0.39 to 0.58, enabling the system to prioritize aggregate throughput and maximize user admission to combat congestion. Notably, fairness ($\beta$) and latency ($\delta$) parameters remain stable around 0.53-0.57 and 0.53-0.57, respectively, preserving equitable resource distribution and latency constraints across load variations.

Fig. \ref{fig:params}(b) demonstrates the framework's responsiveness to heterogeneous traffic bursts, specifically URLLC packet ratios ranging from 0\% to 100\%. As latency-sensitive traffic dominates, $\delta$ escalates from 0.42 to 0.69, reflecting heightened emphasis on minimizing delays. Concurrently, $\gamma$ declines from 0.60 to 0.38, indicating a strategic deprioritization of new admissions to safeguard existing low-latency services. Throughput ($\alpha$) and fairness ($\beta$) exhibit minimal fluctuations (0.64-0.65 and 0.53-0.55), underscoring the framework's ability to dynamically recalibrate objectives without compromising core performance metrics.

\subsection{Effectiveness of the One-Shot OaaS}

\begin{figure*}[!t]
  \centering
  \subfloat{
      \includegraphics[width= 3.3 in]{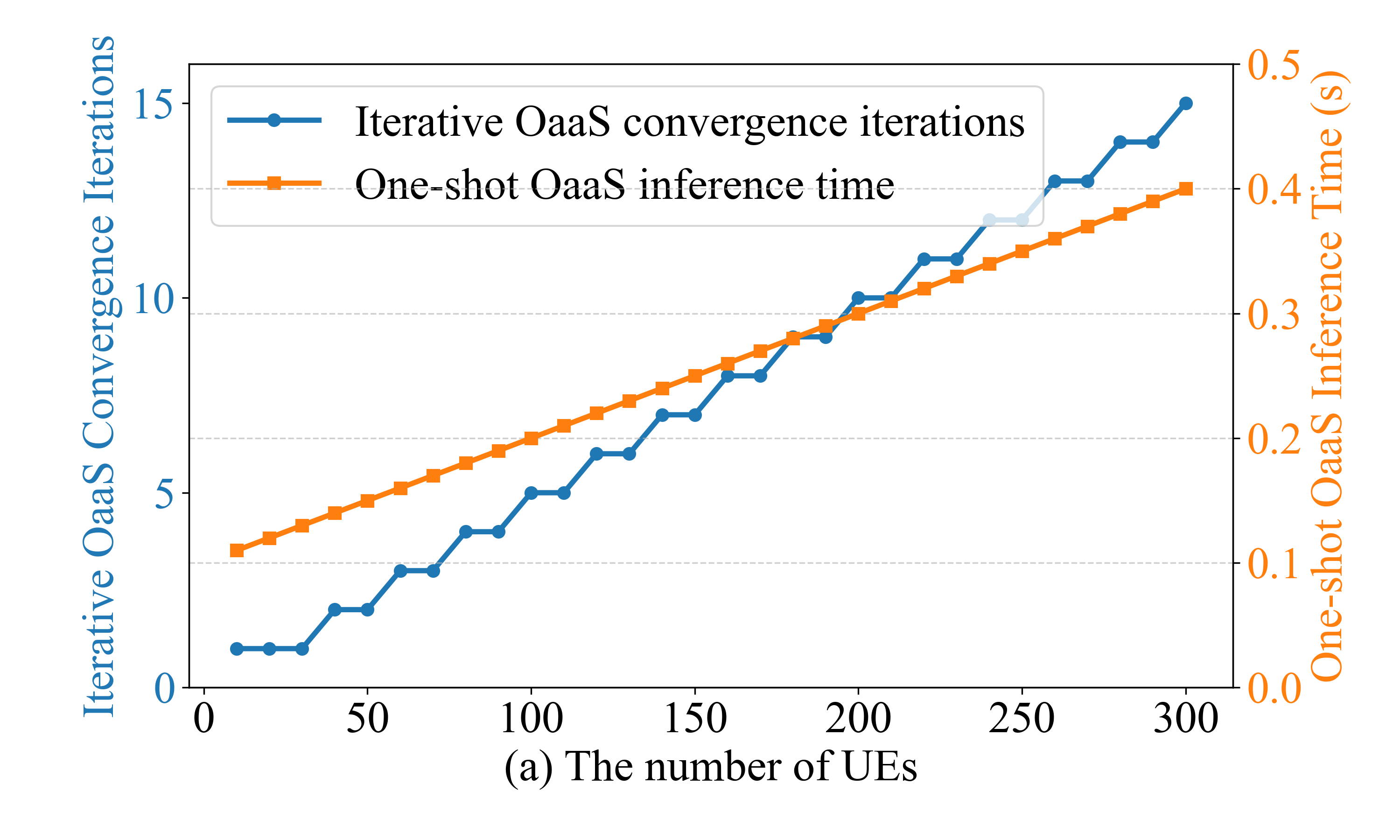}
      \label{fig:iter:a}}
      \hfil
  \subfloat{
      \includegraphics[width=3.3 in]{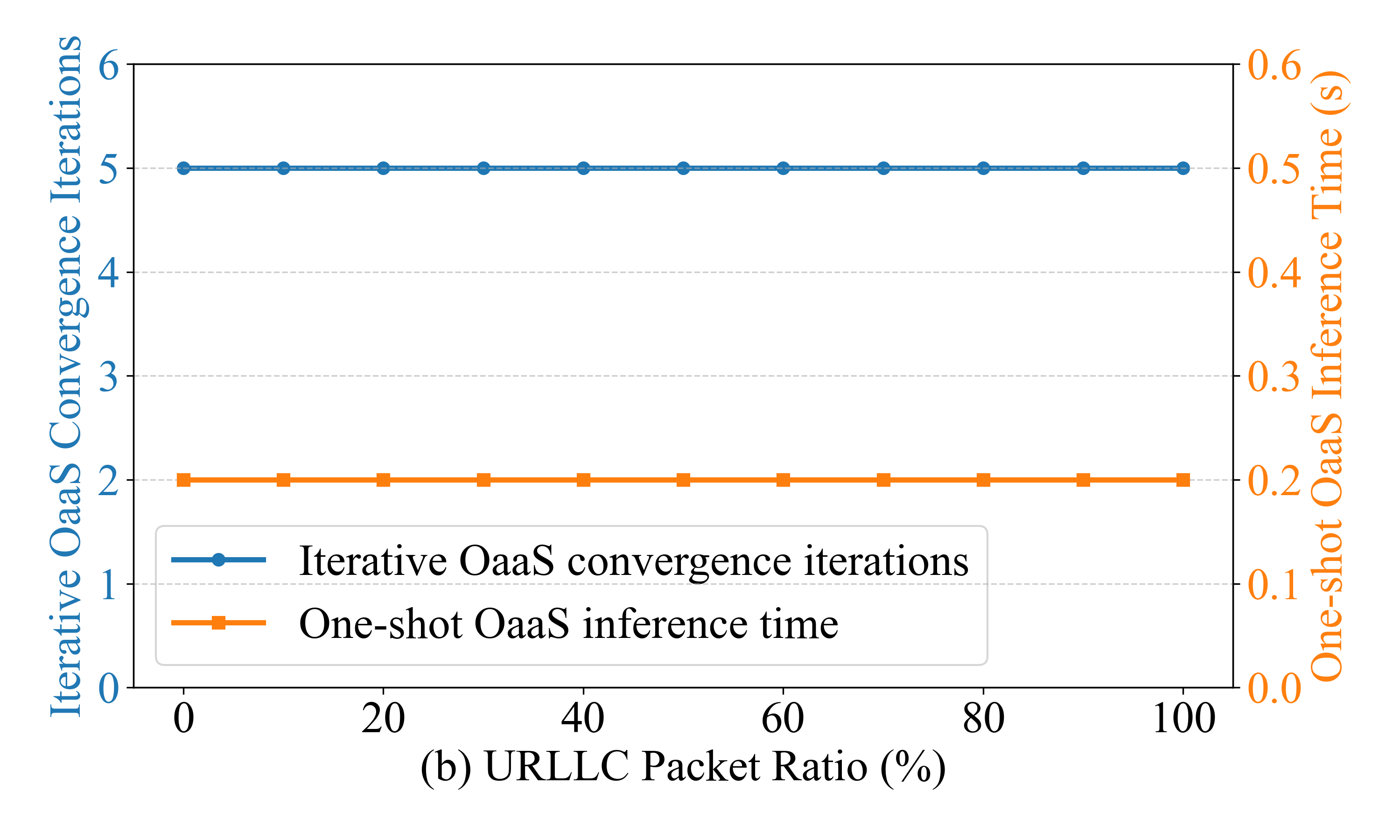}
      \label{fig:iter:b}}
  \caption{The effectiveness of the one-shot OaaS framework, where (a) shows the convergence iterations and one-shot inference time under different numbers of UEs, and (b) shows the convergence iterations and one-shot inference time under different URLLC packet ratios.}
  \label{fig:iter}
\end{figure*}

Fig. \ref{fig:iter}(a) shows the computational behavior of iterative OaaS and one-shot OaaS as the number of UEs increases. The iterative OaaS requires more reflection iterations under heavier network loads, increasing from only a few iterations in sparse scenarios to about 15 iterations when the number of UEs reaches 300. This confirms that the reflection loop becomes more costly as the PRB allocation problem becomes more complex. In contrast, the one-shot OaaS directly predicts the refined objective parameters with a single forward pass, keeping the inference time at a low level even under high user density.

Fig. \ref{fig:iter}(b) further evaluates the impact of heterogeneous URLLC traffic. When the URLLC packet ratio varies from 0\% to 100\%, the iterative OaaS maintains around five convergence iterations, while the one-shot OaaS keeps a stable inference time of about 0.2 s. This indicates that the distilled one-shot model removes the need for online reflection without being sensitive to traffic composition changes. Therefore, one-shot OaaS can significantly reduce online computational overhead while preserving the adaptive objective-generation capability of the original OaaS framework.

\subsection{QoS Performance Analysis}

\begin{figure}[!t]
  \centering
  \subfloat{
      \includegraphics[width=1.6in]{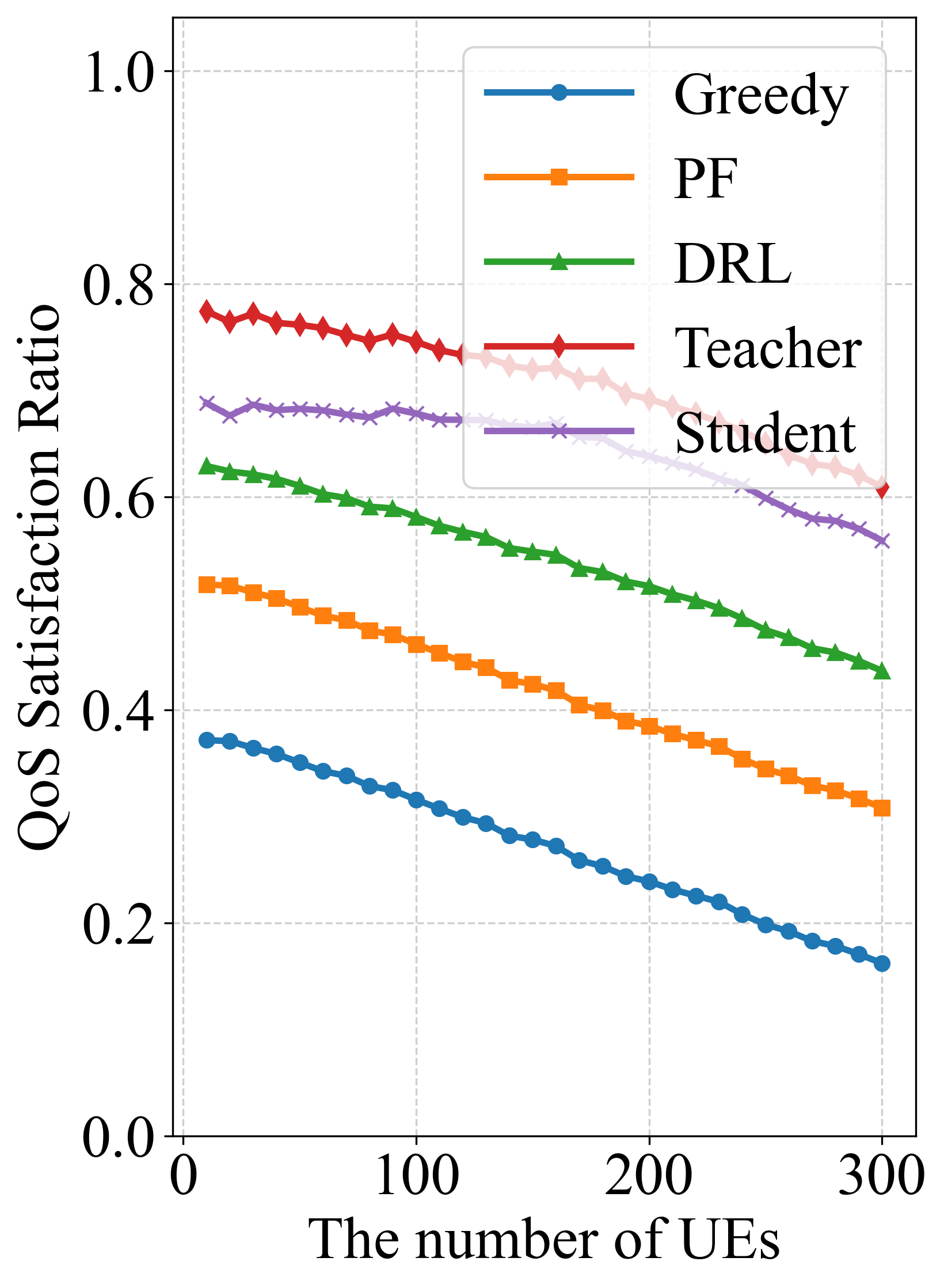}
      \label{fig:qos:a}}
      \hfil
  \subfloat{
      \includegraphics[width=1.6in]{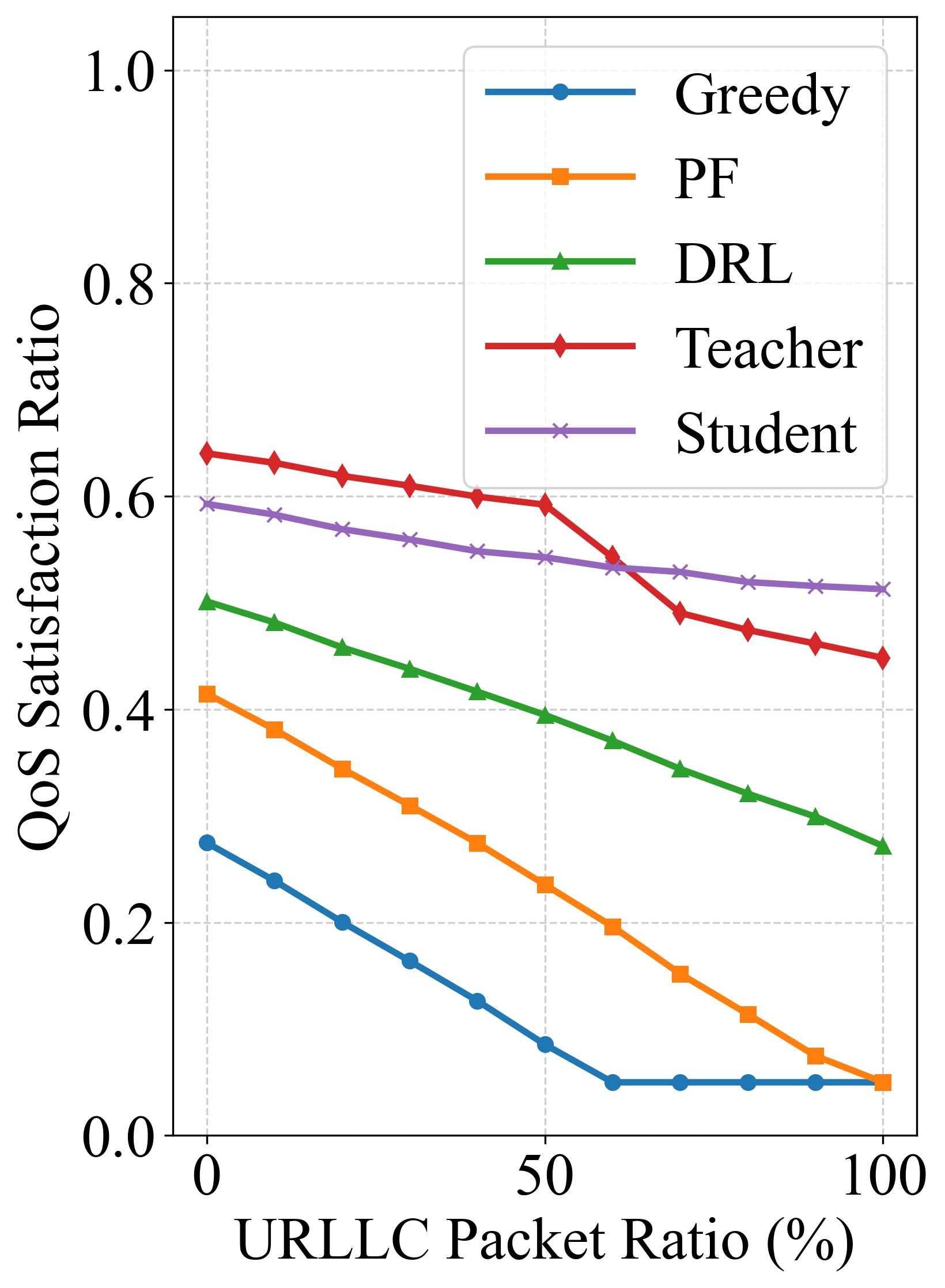}
      \label{fig:qos:b}}
  \caption{QoS performance comparison under different scenarios, where (a) shows the QoS satisfaction ratio under different numbers of UEs, and (b) shows the QoS satisfaction ratio under different URLLC packet ratios with 300 UEs.}
  \label{fig:qos}
\end{figure}

Fig. \ref{fig:qos} compares the QoS satisfaction ratio of different methods under dynamic RAN conditions. In Fig. \ref{fig:qos}(a), the QoS satisfaction ratio decreases as the number of UEs increases, since more users compete for limited PRB resources. Nevertheless, iterative OaaS consistently achieves the best performance, and one-shot OaaS closely follows the teacher model, showing that adaptive objective generation can better handle congestion than the static Greedy, PF, and DRL baselines. In Fig. \ref{fig:qos}(b), increasing the URLLC packet ratio also reduces QoS satisfaction because of stricter latency requirements. The static baselines degrade rapidly under latency-critical traffic, whereas the two OaaS-based methods maintain higher QoS satisfaction by dynamically adjusting the objective toward latency-aware resource allocation.

\section{Conclusions}

In this paper, we have proposed an OaaS framework for adaptive PRB allocation in dynamic 6G RANs. Leveraging a LLM-MA architecture, our framework dynamically constructs context-aware optimization objectives based on real-time network states and heterogeneous QoS/security requirements. To eliminate the online latency of iterative reflection, we have designed a one-shot reflection distillation mechanism where a lightweight student directly predicts refined parameters. We have theoretically bounded the performance gap of this one-shot policy. Simulations confirmed that our framework robustly adapts to volatile conditions, while the one-shot model drastically reduces inference latency with near-optimal QoS satisfaction, establishing OaaS as a flexible and scalable paradigm for future intelligent resource allocation.

{
\footnotesize
\bibliographystyle{IEEEtran}
\bibliography{IEEEabrv,IEEEexample}
}

\end{document}